\documentclass[leqno,11pt,graphicx]{amsart}
\usepackage[colorlinks,linkcolor=blue,anchorcolor=yellow,citecolor=red,urlcolor=green]{hyperref} 
\usepackage{amsmath,amscd,amsthm,amsfonts,appendix,tikz}
\usepackage{amssymb,mathrsfs,xcolor,bbm,dsfont,verbatim}
\usepackage{enumitem}
\usepackage[OT2,OT1]{fontenc}
\usepackage{caption}
\usepackage{longtable}
\usepackage[all]{xy}
\numberwithin{equation}{section}
\allowdisplaybreaks[4]
\theoremstyle{plain}

\definecolor{fadeblue}{RGB}{0,57,128}
\def\fadeblue{\color{fadeblue}}
\usepackage{etoolbox}
\patchcmd{\section}{\normalfont}{\normalfont \fadeblue}{}{}
\patchcmd{\subsection}{\normalfont}{\normalfont \fadeblue}{}{}
\patchcmd{\subsubsection}{\normalfont}{\normalfont \fadeblue}{}{}

\usepackage[normalem]{ulem}

\newcommand{\N}{{\mathbb N}}

\newcommand{\Z}{{\mathbb Z}}

\newcommand{\R}{{\mathbb R}}

\newcommand{\E}{{\mathcal E}}

\newcommand{\F}{{\mathcal F}}

\newcommand{\A}{{\mathcal A}}

\newcommand{\al}{\alpha}
\newcommand{\D}{\mathcal{D}}

\allowdisplaybreaks

\newtheorem{theorem}{Theorem}[section]
\newtheorem{proposition}{Proposition}[section]
\newtheorem{lemma}{Lemma}[section]
\newtheorem{remark}{Remark}[section]

\begin{document}
	\title{Log-Hausdorff multifractality of the absolutely continuous spectral measure of the almost Mathieu operator}

	\author{Jie Cao}
    \address{Chern Institute of Mathematics and LPMC, Nankai University, Tianjin 300071, P. R. China.}
	\email{caojie@nankai.edu.cn}

\author{Xianzhe Li}
\address{Department of Mathematics, University of California, Berkeley, CA 94720, USA} 
 \email{xianzhe@berkeley.edu}
    
 
\author{Baowei Wang}
\address{School of Mathematics and Statistics, Huazhong University of Science and Technology, Wuhan 430074, China} 
\email{bwei\_wang@hust.edu.cn}

 \author{Qi Zhou}
 \address{Chern Institute of Mathematics and LPMC, Nankai University, Tianjin 300071, China}
 \email{qizhou@nankai.edu.cn}


	\begin{abstract}
	This paper focuses on the fractal characteristics of the absolutely continuous spectral measure of the subcritical almost Mathieu operator (AMO)  and Diophantine frequency. In particular, we give a complete description of the (classical) multifractal spectrum and a finer description in the logarithmic gauge. The proof combines continued–fraction/metric Diophantine techniques and refined covering arguments. These results rigorously substantiate (and quantify in a refined gauge) the physicists' intuition that the absolutely continuous component of the spectrum is dominated by energies with trivial scaling index, while also exhibiting nontrivial exceptional sets which are negligible for classical Hausdorff measure but large at the logarithmic scale.

    \end{abstract}
	
	\maketitle	
	
	\section{Introduction}
	In the 1980s, there emerged an almost-periodic flu in the study of the Sch\"odinger operator, which sweep the world \cite{simon1982almost}. The most extensively studied model is the almost Mathieu operator (also known as the Aubry-Andr\'e model in the physical literature):
	\begin{equation*}\label{AMO}
		(H_{\lambda,\alpha,\theta}u)(n) = u(n-1) + u(n+1) + 2\lambda\cos(2\pi(n\alpha + \theta))u(n),
	\end{equation*} where $\lambda \in \mathbb{R}$ represents the coupling, $\alpha \in \mathbb{R}$ denotes the frequency (typically irrational), and $\theta \in \mathbb{R}$ is the phase. The spectrum $\Sigma_{\lambda,\alpha}$ is a compact subset of $\mathbb{R}$ that is independent of $\theta$. The almost Mathieu operator (AMO), named by B. Simon, models electrons on a two-dimensional lattice subjected to a perpendicular magnetic field \cite{P1933}. It is of significant interest due to both its physical relevance and the remarkable complexity of its associated spectral theory \cite{aj,ayz,last1995almost,last2005spectral,rauh1974degeneracy,TKNN}.
	
	Indeed, the AMO has been most notably studied for its fractal spectrum, famously visualized as Hofstadter's butterfly. It was popularized by Simon as the ``Ten Martini Problem" \cite{Kac,simon}, which asserts that $H_{\lambda,\alpha,\theta}$ possesses a Cantor spectrum for any $\lambda \neq 0$ and $\alpha \in \mathbb{R} \setminus \mathbb{Q}$. This assertion was ultimately proven by Avila and Jitomirskaya \cite{aj}, with additional contributions from earlier studies \cite{ak,bellisimon,cey,hs,last,puig}. However, it remains an open question whether the ``Dry Ten Martini Problem" (which posits that all spectral gaps are open) holds \cite{ayz2,Kac,simon}. Another significant problem is determining the Hausdorff dimension of the spectrum of the critical almost Mathieu operator (where $\lambda = 1$). B. Simon included this problem in his new list of significant unsolved problems \cite{FiftyyearSimon}, and recent advances on this topic can be consulted in \cite{ALSZ2024abominable,HLQZ2019Positive,HLQZ2025,JKCritical}.
	
	The aforementioned results pertain to the fractal nature of the spectrum. Recently, there has been a growing interest in exploring the fractal characteristics of the spectral measure, which constitutes the primary focus of this paper. For the case where $ \lambda > 1 $,  $H_{\lambda,\alpha,\theta}$  has pure point spectrum for a.e. $\alpha$ and a.e. $\theta$ \cite{jitomirskaya1999metal}, and recently Jitomirskaya and Liu \cite{jitomirskaya2018universal,jLiu1} investigate the universal hierarchical structure of quasiperiodic eigenfunctions. Conversely, when $ \lambda < 1 $, the spectrum is purely absolutely continuous \cite{avila2008absolutely}, with additional earlier contributions \cite{eliasson,aj1,AD2008Absolute}.   Motivated by the conjecture of Tang and Kohmoto \cite{TK1986Global}, the precise local distribution of this absolutely continuous spectral measure has been studied recently \cite{LYZ}.
	
	To elucidate this, let $ \mu $ be a compactly supported Borel probability measure on $ \mathbb{R} $. For a given $ x \in \mathbb{R} $, the lower and upper local dimensions of $ \mu $ at $ x $ are defined as follows:
	\begin{equation}\label{def-loc-dim}
		\underline{d}_{\mu}(x) := \liminf_{r \to 0} \frac{\log \mu(B(x,r))}{\log r}, \quad \overline{d}_{\mu}(x) := \limsup_{r \to 0} \frac{\log \mu(B(x,r))}{\log r},
	\end{equation}
	where $ B(x,r) $ denotes the closed ball in $ \mathbb{R} $ with radius $ r $ centered at $ x $. If $ \underline{d}_{\mu}(x) = \overline{d}_{\mu}(x) := d_{\mu}(x) $, then $ d_{\mu}(x) $ is referred to as the local dimension (or scaling index) of $ \mu $ at $ x $.	
	With this framework established, we can address the conjecture. Denote by $\mu = \mu_{\lambda, \alpha, \theta}$ the spectral measure corresponds to $ H_{\lambda,\alpha,\theta} $. In 1986, physicists Tang and Kohmoto \cite{TK1986Global} conjectured for the absolutely continuous spectrum (extended states) that:
	\begin{quote}
		``An absolutely continuous spectrum is dominated by points with a `trivial' scaling index $ d_{\mu}(E) = 1 $ and a fractional dimension of 1. It may contain a finite or countably infinite number of singularities with $ d_{\mu}(E) \neq 1 $, possibly van Hove singularities.''
	\end{quote}
	
	Let $\mathcal{N}(E) := n\big((-\infty, E]\big)$ denote the integrated density of states (IDS), where $n = \int_{\theta} \mu_{\lambda, \alpha, \theta} \, d\theta$ is the density of states measure. Li, You and Zhou \cite{LYZ} proved that if $\lambda<1$ and $ \alpha \in \mathrm{DC} := \bigcup_{\gamma > 0,\, \tau > 1} \mathrm{DC}(\gamma, \tau) $ is Diophantine, where $$
	\mathrm{DC}(\gamma, \tau) := \left\{ x \in \mathbb{R} : \|kx\|_{\mathbb{R}/\mathbb{Z}} \geq \frac{\gamma}{|k|^\tau} \ \text{for all } k \in \mathbb{Z} \setminus \{0\} \right\},
	$$ the following results hold:
	\begin{itemize}
		\item If $ \mathcal{N}(E)= k\alpha\mod \Z $, then $\underline{d}_{\mu}(E) = \overline{d}_{\mu}(E)=\frac{1}{2}$.
		\item If $ \mathcal{N}(E)\neq  k\alpha\mod \Z $, then
		$$\underline{d}_{\mu}(E) \in [1/2,1], \quad \overline{d}_{\mu}(E)=1. $$
	\end{itemize}

    Taken together, these results suggest that the absolutely continuous   spectral measure \(\mu\) is governed by the behaviour of its lower pointwise (local) dimension.  To describe this local behaviour precisely we study the level sets
\[
\Sigma_{\lambda,\alpha}(\beta)
:=\bigl\{E\in\Sigma_{\lambda,\alpha}:\ \underline d_\mu(E)=\beta\bigr\}.
\]
Multifractal analysis provides the natural language for this study: the family \(\{\Sigma_{\lambda,\alpha}(\beta)\}_{\beta}\) decomposes the spectrum according to local dimensions and reveals the fractal geometry of the exceptional energies where \(\mu\) deviates from the typical  trivial scaling.  In particular, measuring the size of these level sets (via Hausdorff or logarithmic–Hausdorff measures) quantifies how common each type of local behaviour is.

	\subsection{Multifractal formalism}
Historically, multifractal analysis developed from physicists' heuristics into a rigorous mathematical discipline.  The first influential, systematic exposition is usually attributed to Mandelbrot \cite{M1974}, where he proposed that the bulk of energy dissipation in turbulent flows is concentrated on a subset of \(\mathbb R^3\) of fractional dimension.  Since that seminal work, multifractal ideas have been widely pursued in both physics and mathematics.  Modern multifractal theory provides a compact language and robust tools—local dimensions, the multifractal spectrum, mass–distribution methods, thermodynamic formalism and ubiquity techniques—to dissect measures with nonuniform scaling.  These methods have found applications across turbulence, dynamical systems, geometric measure theory, signal analysis and spectral theory; conversely, empirical problems continue to motivate new rigorous developments.  For comprehensive introductions and further reading we refer the reader to the monographs \cite{BPPV1984,Fal2014Fractal,MCS1986,FP1985}.

	In our  context,  the multifractal analysis aims to examine the multifractal spectrum (or the $ f_\mu(\beta) $-spectrum) of the measure $ \mu $ \cite{Pesin}, defined by
	$$
	f_\mu(\beta) = \dim_H(\Sigma_{\lambda, \alpha}(\beta)),
	$$ where $ \dim_H(S) $ denotes the Hausdorff dimension of the set $ S $. This analysis is often sufficient to reveal the underlying fine structure of the measure (see for examples \cite{Bowen,Cawley,Lau,Olsen,Pesin1} for classic results).  Referring back to Tang and Kohmoto's initial conjecture \cite{TK1986Global}, we are particularly interested in determining the dimension of the set of energies $ E $ for which $ d_{\mu}(E) \neq 1 $. This paper aims to address these questions.

	\begin{theorem}\label{thm:stdhausdorff}
		Assume that $0<\lambda<1$ and $\alpha\in DC$. Let $ f_\mu(\beta) $ be defined as above; then we have
		\[
		f_\mu(\beta) =\begin{cases}
			1, \quad&\beta=1,\\
			0, \quad&\beta\in[1/2,1).
		\end{cases}
		\]
	\end{theorem}

	Theorem \ref{thm:stdhausdorff} demonstrates that the classical power-law gauge function is insufficiently precise to differentiate between the level sets $ \Sigma_{\lambda, \alpha}(\beta)$ for $ \beta \in [1/2, 1) $. To capture the finer structure of these sets, it is essential to adopt a more refined gauge function.
	
	Let $ \omega: [0, 1] \to [0, \infty) $ be defined as a gauge function, which is an increasing function satisfying $ \omega(0) = 0 $. Typical examples include the classical power-law gauge function $ r^s $ and $$
	\omega_s(r) =
	\begin{cases}
		(-\log r)^{-s}, & \text{if } 0 < r <1, \\
		0, & \text{if } r = 0.
	\end{cases}
	$$ The $ \omega $-Hausdorff measure $ \mathcal{H}^{\omega}(S) $ of a set $ S \subset \mathbb{R} $ is defined as $$
	\mathcal{H}^{\omega}(S) = \lim_{\varepsilon \to 0^+} \mathcal{H}_{\varepsilon}^{\omega}(S),
	$$ where $$
	\mathcal{H}^{\omega}_{\varepsilon}(S) := \inf \left\{ \sum_{i=1}^{\infty} \omega(b_i - a_i) : S \subset \bigcup_{i=1}^{\infty} (a_i, b_i), \ b_i - a_i \leq \varepsilon \right\},
	$$ while the $\log$-Hausdorff dimension is defined as $$
		\dim_{H, \log}(S):=\inf\{s>0: \mathcal{H}^{\omega_s}(S)<\infty\}.
		$$
        
        The concept of $\log$-Hausdorff dimension is particularly interesting in the context of Schr\"odinger operators. It is established by Bourgain-Klein \cite{BourgainKlein2013} and  Craig-Simon  \cite{CraigSimon1983}, that the spectrum of any one-dimensional discrete Schr\"odinger operator possesses a positive $\omega_1$-Hausdorff measure. Furthermore, Avila-Last-Shamis-Zhou \cite{ALSZ2024abominable} demonstrate that this result cannot be improved, even for the almost Mathieu operator. Additional abominable properties of the almost Mathieu operator can be explored within the framework of the $\omega_s$ gauge category \cite{ALSZ2024abominable}. In this paper, we determine the $w_s$-Hausdorff measure of the level set $\Sigma_{\lambda, \alpha}(\beta)$ by showing the following zero-infinity dichotomy.

	\begin{theorem}\label{DtoF}
		Let $0<\lambda<1$ and $\alpha\in DC$. For any $ \beta \in [1/2, 1) $, $$
		\mathcal{H}^{\omega_s}(\Sigma_{\lambda, \alpha}(\beta))=
		\begin{cases}
			0, & \hbox{for $s>1$,} \\
			\infty, & \hbox{for $s\le 1$.}
		\end{cases}
		$$ and so $\dim_{H, \log}(\Sigma_{\lambda, \alpha}(\beta))=1.$ 
	
	\end{theorem}
\begin{remark}  
     Intuitively, classical power–law gauges average away the contribution of rare, strong resonances which occur at logarithmically small scales; the logarithmic gauge $\omega_s(r)=(-\log r)^{-s}$ is finely tuned to detect and measure these resonances, which is why it reveals a nontrivial multifractal structure invisible to power–law analysis.
\end{remark}

	\subsection{Ideas of the proof and reduction to a Diophantine approximation problem}

The analysis proceeds by linking lower dimension $\underline{d}_{\mu}$  of the spectral measure to the strength of resonances of the IDS $\mathcal{N}(E)$.  Two input facts are essential:

\begin{enumerate}
  \item  H\"older regularity of the IDS.  This regularity controls how small neighborhoods of an energy can be while still capturing a prescribed amount of density.
  \item A quantitative relation between the resonance strength of $\mathcal{N}$ along the frequency orbit $\{k\alpha\}$ and the lower local dimension of the spectral measure $\mu$.  Roughly speaking, large near-resonances of $\mathcal{N}$ at energies $E$ induced by rational approximations of $\alpha$ force $\mu(B(E,r))$ to be smaller than a pure power of $r$, producing reduced lower local dimension.
\end{enumerate}
Using these two ingredients we reduce the multifractal analysis to a purely Diophantine approximation problem for the circle. Let's explain this in details. 
    
	 For any $\varphi \in [0,1]$, define the resonance strength
	\begin{equation*}\label{def-delta}
		\delta(\alpha, \varphi) = \limsup_{|k| \to \infty} -\frac{\log \| \varphi - k\alpha \|_{\mathbb{R}/\mathbb{Z}}}{|k|},
	\end{equation*}
	with the convention that $\delta(\alpha, \varphi) = \infty$ if $\varphi \equiv k\alpha \mod \mathbb{Z}$ for some $k \in \mathbb{Z}$.
 For any $0 < \delta \leq \infty$, define the level set of $\Sigma_{\lambda,\alpha}$ as
	\[
	F(\delta) := \left\{ E \in\Sigma_{\lambda,\alpha} : \delta(\alpha, \mathcal{N}(E)) = \delta \right\}.
	\]
   As proved by  \cite{LYZ}
	for $\beta \in [1/2, 1)$, the lower level set $\Sigma_{\lambda, \alpha}(\beta)$ satisfies
	\begin{equation}\label{Xmu-beta}
		\Sigma_{\lambda, \alpha}(\beta) = F\left( \frac{\beta \log \lambda}{1 - 2\beta} \right).
	\end{equation}
	This motivates the study of the auxiliary set
	\begin{equation*}\label{D-def}
		D(\delta) := \left\{ x \in [0,1] : \delta(\alpha, x) = \delta \right\} = \Big\{ x \in [0,1] : \limsup_{|k| \to \infty} -\frac{\log \| x - k\alpha \|_{\mathbb{R}/\mathbb{Z}}}{|k|} = \delta \Big\},
	\end{equation*}
	which satisfies $D(\delta) = \mathcal{N}(F(\delta))$ since $\mathcal{N}$ is a continuous, non-decreasing surjective function.
	
We also recall a well-known from fractal geometry.
	For the power-law gauge function $\omega$ and a bi-Lipschitz function $f : \mathbb{R} \to \mathbb{R}$ (i.e., there exists $s > 0$ such that for all $x, y \in \mathbb{R}$,
	$
	s^{-1} |x - y| \leq |f(x) - f(y)| \leq s |x - y|
	$),
	the Hausdorff dimension satisfies $\dim_H(S) = \dim_H(f(S))$ for any $S \subset \mathbb{R}$ \cite{Fal2014Fractal}. However, the integrated density of states $\mathcal{N}$ is generally only H\"older continuous:
	\begin{itemize}
		\item In the zero Lyapunov exponent regime (small $\lambda$), see \cite{aj1, CCYZ, eliasson,amor}.
		\item In the positive Lyapunov exponent regime (large $\lambda$), see \cite{bourgain2000holder, goldstein2008fine,Klein}.
	\end{itemize}
	While uniform H\"older lower bounds are unavailable \cite{KXZ2020AnosovKatok,LYZ}, a weak lower bound exists (Proposition \ref{N-continuous}). This allows us to relate the $\omega_s$-Hausdorff measures of $D(\delta)$ and $F(\delta)$:
	
	\begin{proposition}\label{DtoF'}
		Let $0<\lambda<1$ and $\alpha\in DC$. For any $0 < \delta \leq \infty$ and $s>0$,
		\begin{equation}\label{D-F}
			3^{-s}\cdot\mathcal{H}^{\omega_s}(D(\delta)) \leq \mathcal{H}^{\omega_s}(F(\delta)) \leq 3^{s+1}\cdot\mathcal{H}^{\omega_s}(D(\delta)).
		\end{equation}
	\end{proposition}
	
	\begin{remark}
		In fact, one can obtain the following result: if there exists $t > 2$ such that the gauge function $\omega$ satisfies
		\begin{equation*}
			\limsup_{r \to 0^+} \frac{\omega(r)}{\omega(r^t)} < K < \infty,
		\end{equation*}
		then $K^{-1}\cdot\mathcal{H}^{\omega}(D(\delta)) \leq \mathcal{H}^{\omega}(F(\delta)) \leq 2K\cdot\mathcal{H}^{\omega}(D(\delta)).$
	\end{remark}

	Proposition \ref{DtoF'} reduces the problem to analyzing the Diophantine set $D(\delta)$:
	
	\begin{theorem}\label{F}
		For any irrational number $\alpha\in[0,1]$ and $0<\delta\leq\infty$, we have $$
		\mathcal{H}^{\omega_s}(D(\delta))=
		\begin{cases}
			0, & \hbox{if $s>1$,} \\
			\infty, & \hbox{if $s\le 1$.}
		\end{cases}
		$$ So we have $\dim_{H, \log}D(\delta)=1$.
	\end{theorem}

\begin{remark}
Theorem~\ref{F} holds for every irrational frequency \(\alpha\), whereas Theorem~\ref{DtoF} is stated only for Diophantine \(\alpha\).  The distinction stems from the dependence of the reduction on Proposition~\ref{N-continuous}: this proposition provides the uniform regularity of the integrated density of states (the continuity/Hölder estimates used in the bridge argument Proposition \ref{DtoF'}) and, as shown in \cite{ALSZ2024abominable}, is available only for Diophantine frequencies. 
\end{remark}

Indeed, both \(D(\delta)\) and the classical limsup set
\[
E(\psi):=\{y\in[0,1]:\ \|y-k\alpha\|_{\R/\Z}<\psi(k)\ \text{for infinitely many }k\in\mathbb Z\}
\]
belong to the same family of Diophantine–approximation problems: they measure how well points on the circle are approximated by the orbit \(\{k\alpha\}\).  From this point of view the set \(E(\psi)\) is a shrinking–target  set for the rotation, and its metric size (Lebesgue measure, Hausdorff measure and Hausdorff dimension) has been a central object of study since the pioneering work of Kurzweil \cite{K1955}.  Over the decades many authors investigated variants and refinements of Kurzweil's problem (see, e.g., \cite{BD1999,B2003,LR2013,TS2003}).  Only in recent years, however, have complete results been obtained describing both the Lebesgue measure and the dimension–theoretic behaviour of the sets \(E(\psi)\) for general approximating functions \(\psi\); see \cite{FK2016,KRW2018} for modern treatments and precise statements.

The set \(D(\delta)\) is more delicate because a point \(x\in D(\delta)\) must satisfy two simultaneous constraints: (i) it is very well approximated by \(k\alpha\) along infinitely many indices (as in \(E(\psi)\)), and (ii) it remains quantitatively separated  for all orbits.  To achieve both requirements we must control how closely and how often the orbit points \(k\alpha\) come together.  Continued fractions are the standard tool for this: their convergents \(p_n/q_n\) describe natural scales at which the orbit nearly repeats, so by using those scales we can pick a sparse subsequence of indices that produce very close returns while keeping the other orbit points uniformly separated.

The proof of Theorem~\ref{F} borrows the constructive flavor of classical metric methods (cf.\ \cite{BGN,Bugeaud}) where their proof is based on the fact that the resonant points form a ubiquitous systems \cite{BDV}. However, for irrational rotation, $\{k\alpha: k\ge 1\}$ is not a ubiquitous system which is also the main reason why the measure and dimension of $E(\psi)$ was not solved over a long time. Our proof departs from them in two key ways: (a) we work in a logarithmic gauge rather than a power–law one, and (b) membership in \(D(\delta)\) imposes both upper and lower asymptotic constraints simultaneously, which complicates overlap and multiplicity estimates.  To handle these issues we combine (i) continued–fraction separation lemmata, (ii) a covering strategy adapted to the logarithmic gauge, and (iii) a tailored mass–distribution construction on Cantor–type sets.  This short, three-part synthesis is the arithmetic core that underpins our estimates for \(D(\delta)\) and, via the spectral reduction, yields the multifractal results in the paper.

	\section{Proof of Proposition \ref{DtoF'}}
	
	In this section, we establish the Hausdorff measure-transitivity property for the sets $D(\delta)$ and $F(\delta)$ under the gauge function $\omega_s$.
	
			
			

	\begin{proof}[Proof of Proposition \ref{DtoF'}]
		
		Fix any $0<\delta\leq\infty$. Since the integrated density of states $\mathcal{N}:\Sigma_{\lambda, \alpha}\to[0,1]$ is a continuous non-decreasing surjective function, by the definitions of $D(\delta)$ and $F(\delta)$, it can be checked directly that
		$$D({\delta})=\mathcal{N}(F({\delta})).$$

		Recall that the integrated density of states $\mathcal{N}$ is uniformly $ 1/2 $-H\"older continuous. More precisely, we have the following proposition:
		
		\begin{proposition}[\cite{avila2008absolutely}]\label{N-continuous}
			Let $ \alpha \in \mathrm{DC} $ and $ 0 < \lambda < 1 $. Then, there exists $ 0<c = c(\lambda,\alpha)<1 $ such that for any $ E \in \Sigma_{\lambda,\alpha} $ and $ 0 < \varepsilon < 1 $, the following inequality holds:
			\[
			c\,\varepsilon^{\frac{3}{2}} \leq \mathcal{N}(E + \varepsilon) - \mathcal{N}(E - \varepsilon) \leq c^{-1}\, \varepsilon^{\frac{1}{2}}.
			\]
		\end{proposition}
		
		Next we prove the first inequality of \eqref{D-F}. Given any countable cover $(I_i)_i$ of $F(\delta)$ consisting of closed intervals $I_i=[E_i^1,E_i^2]$ with $|I_i|:=\text{diam}\ I_i\leq c^{6}$.
		By Proposition \ref{N-continuous} and $\omega_s$ is increasing, then
		\begin{equation*}
			\omega_s(|\mathcal{N}(I_i)|)=\omega_s\left(\mathcal{N}(E_{i}^{2})-\mathcal{N}(E_{i}^{1})\right)\leq \omega_s(c^{-1}|I_i|^{\frac{1}{2}})\leq3^s\omega_s(|I_i|),
		\end{equation*}
		for the last inequality as above, we use $|I_i|\leq c^{6}$. Since $D(\delta)=\mathcal{N}(F(\delta))\subset\bigcup_i\mathcal{N}(I_i)$, then $\bigcup_i\mathcal{N}(I_i)$ forms a cover of $D(\delta)$, we have
		\begin{align*}\label{D=F-1}
			\mathcal{H}^{\omega_s}(D(\delta))\leq	\sum_{i}\omega_s\left(|\mathcal{N}(I_i)|\right)&\leq 3^s\sum_i \omega_s(|I_i|).
		\end{align*}
		This gives $3^{-s}\mathcal{H}^{\omega_s}(D(\delta))\leq \mathcal{H}^{\omega_s}(F(\delta)) $.
		
		At last we prove the second inequality of \eqref{D-F}. Fix countable cover $(J_i)_i$ of $D(\delta)$ consisting of closed intervals $J_i=[a_i,b_i]$ with $|J_i|=|b_i-a_i|\leq(\frac{c}{6})^2$ and $|\mathcal{N}^{-1}(J_i)|<1/2$. Noting that  $ \mathcal{N}$  is not bijective, and it is locally constant in the resolvent set $\R \backslash \Sigma_{\lambda, \alpha},$ thus we can take
		\begin{equation*}
			E_i^1=\max\{E\in\Sigma_{\lambda, \alpha}:\mathcal{N}(E)\leq a_i\}\ \ \ \ \text{and}\ \ \ \ E_i^2=\min\{E\in\Sigma_{\lambda, \alpha}:\mathcal{N}(E)\geq b_i\}.
		\end{equation*}
		It follows that $E_{i}^{1},E_{i}^{2}\in\Sigma_{\lambda,\alpha}$,
		and denote $I_i=[E_i^1,E_i^2]$, then we have
		$\mathcal{N}(I_i)=[\mathcal{N}(E_{i}^{1}),\mathcal{N}(E_{i}^{2})]=J_i$ and $|I_i|<1/2$.
		
		In the following, we construct a cover $(I_i^1\cup I_i^2)_i$ of $F(\delta)$ with $(I_i^1\cup I_i^2)\subset I_i$ such that
		\begin{equation}\label{sum-J_i'}
			|\mathcal{N}(I^1_i)|\geq|I^1_i|^{3}\ \ \ \ \text{and}\ \ \ \ |\mathcal{N}(I^2_i)|\geq|I^2_i|^{3}.
		\end{equation}	
		
		We distinguish the proof into two cases:

		\noindent\textbf{Case 1: If there exists $ t_0\in[1/3,2/3]  $ such that
			$
			E_i^*:=E_{i}^{1}+t_0(E_{i}^{2}-E_{i}^{1})\in\Sigma_{\lambda,\alpha}
			$.}
		In this case, 	let \begin{equation}\label{1}
			I^1_i=I^2_i:=[E_{i}^{1}, E_{i}^{2}].
		\end{equation}
		Without loss of generality, we only to consider $t_0\in[1/2,2/3]$. Now we see that $$0<E_i^2-E_i^*=(1-t_0)|E_i^2-E_i^1|<|I_i|<1,$$ then
		by Proposition \ref{N-continuous} and monotonicity of $\mathcal{N}$, we have
		\begin{align}\label{N-continuous-1}
			|\mathcal{N}(I^1_i)|=|\mathcal{N}(I^2_i)|
			&\geq|\mathcal{N}(E_i^*+(E_{i}^{2}-E_i^*))-\mathcal{N}(E_i^*-(E_{i}^{2}-E_i^*))|\notag\\
			&\geq c|E_i^2-E_i^*|^{\frac{3}{2}}\geq c(1-t_0)^{\frac{3}{2}}|E_i^2-E_i^1|^{\frac{3}{2}}\notag\\
			&\geq\frac{c}{6}|E_{i}^{2}-E_{i}^{1}|^{\frac{3}{2}}\geq|E_i^2-E_i^1|^{3}\\&=|I^1_i|^{3}=|I^2_i|^{3},\nonumber
		\end{align}
		where, in the last inequality of \eqref{N-continuous-1}, we use  $$|E_i^2-E_i^1|\leq(\frac{6}{c}|\mathcal{N}(I_i)|)^{2/3}=(\frac{6}{c}|J_i|)^{2/3}\leq(\frac{c}{6})^{\frac{2}{3}}.$$

		\noindent\textbf{Case 2:  If there is no $ t_0\in[1/3,2/3] $ such that
			$
			E_{i}^{1}+t_0(E_{i}^{2}-E_{i}^{1})\in\Sigma_{\lambda,\alpha}
			$.}
		This means there is a spectral gap $$ G_{k}=[E_i^-,E_i^+]\subset[E_{i}^{1},E_{i}^{2}]$$ with $E_i^-,E_i^+\in\Sigma_{\lambda,\alpha}$ such that $ \mathcal{N}(E_i^-)=\mathcal{N}(E_i^+)=\mathcal{N}(E)=k\alpha \mod\Z $ for all $E\in G_k$. In this case, let
		\begin{equation}\label{2}
			I^1_i:=[E_{i}^{1},E_i^-]\ \ \ \text{and}\ \ \ I^2_i:=[E_i^+,E_{i}^{2}].
		\end{equation}
		For the interval $I^1_i$, note that $ E_i^{-}\in \Sigma_{\lambda,\alpha} $ and $ \mathcal{N}(E_i^-+E_i^--E_{i}^{1})=\mathcal{N}(E_i^-)$ (since $2E_i^--E_{i}^{1}\in G_k$), by Proposition \ref{N-continuous}, we have
		\begin{align}\label{N-continuous-2}
			|\mathcal{N}(I^1_i)|&=|\mathcal{N}(E_i^-+E_i^--E_{i}^{1})-\mathcal{N}(E_i^--(E_i^--E_{i}^{1}))|\geq c|E_i^--E_{i}^{1}|^{\frac{3}{2}}\geq|I^1_i|^{3}.
		\end{align}
		Similarly, for the interval $I^2_i$, we have
		\begin{equation}\label{N-continuous-3}
			|\mathcal{N}(I^2_i)|=|\mathcal{N}(E_{i}^{2})-\mathcal{N}(E_i^+)|\geq |I^2_i|^{3}.
		\end{equation}
		Thus, we finish the constructions of $I^1_i$, $I^2_i$, and $(I_i^1\cup I_i^2)\subset I_i$ by \eqref{1} and \eqref{2}.
		
		As \eqref{sum-J_i'}  follows from \eqref{N-continuous-1},\eqref{N-continuous-2} and \eqref{N-continuous-3}.
		We  are  left to show that $(I_i^1\cup I_i^2)_i$ is a cover of $F(\delta)$. By the definition of $I^1_i$ and $I^2_i$, we have $$\mathcal{N}(I_i^1)\cup\mathcal{N}(I_i^2)=\mathcal{N}(I_i)=J_i,$$ and hence $$\mathcal{N}^{-1}(J_i)\subset I_i^1\cup I_i^2.$$
		Since $D(\delta)=\mathcal{N}(F(\delta))$ and $(J_i)_i$ is a cover of $D(\delta)$, then $$F(\delta)\subseteq \mathcal{N}^{-1}(D(\delta))\subseteq\bigcup_i\mathcal{N}^{-1}(J_i)\subset\bigcup_i(I^1_i\cup I^2_i)_i.$$
		
		Since $(I_i^1\cup I_i^2)_i$ is a cover of $F(\delta)$, then by \eqref{sum-J_i'}
		\begin{align*}
			\mathcal{H}^{\omega_s}(F(\delta))&\leq	 \sum_i\omega_s(|I^1_i|+|I^2_i|)\leq\sum_{i}\omega_s\left(|\mathcal{N}(I^1_i)|^{\frac{1}{3}}+|\mathcal{N}(I^2_i)|^{\frac{1}{3}}\right)\\
			&\leq2\sum_{i}\omega_s(|\mathcal{N}(I_i)|^{\frac{1}{3}})=2\sum_{i}\omega_s(|J_i|^{\frac{1}{3}})\leq 2\cdot 3^s\sum_i\omega_s(|J_i|).
		\end{align*}
		Therefore, we have $\mathcal{H}^{\omega_s}(F(\delta))\leq3^{s+1}\mathcal{H}^{\omega_s}(D(\delta))$.
	\end{proof}

	\section{Proof of Theorem \ref{F}}
	
	The crux of proving Theorem \ref{F} lies in identifying a Cantor subset
	$C(\delta)\subset D(\delta)$ that facilitates the estimation of the Hausdorff measure of $D(\delta)$. This section is organized into three main parts:
	
	First, we review the mass distribution principle and the distribution properties of irrational numbers in Section \ref{3-1}; we then establish two fundamental propositions used in the construction of the Cantor subset
	$C(\delta)$ in Section \ref{3-2}.
	
	Next, we construct the Cantor subset $C(\delta)\subset D(\delta)$ in Section \ref{3-3}; Section \ref{3-4} assigns a mass distribution to $C(\delta)$; and Section \ref{3-5} derives the Hausdorff measure of $C(\delta)$.
	
	Finally, the proof of Theorem \ref{F} is completed in Section \ref{3-6}.
	
	\subsection{Preliminaries}\label{3-1}
	
	We cite the mass distribution principle which is a classic tool to determine the Hausdorff dimension and Hausdorff measure of a set from below.
	\begin{lemma}[\cite{F}] \label{Fal}
		Let $S\subseteq \mathbb{R}$ be a Borel set and $\mu$ be a Borel measure with $\mu(S)>0$.  Let $\omega$ be a dimension function. If there exist $c>0$ and $r_0>0$ such that for any $x\in S$ and $r\le r_0$,
		\begin{equation*}
			\mu(B(x,r))\le c\cdot \omega(r),
		\end{equation*}
		where $B(x,r)$ denotes the ball with center $x$ and radius $r$, then $\mathcal{H}^{\omega}(S)\ge \mu(S)/c>0.$
	\end{lemma}
	
	The next result concerns the distribution of $\{k\alpha: k\ge 1\}$ for an irrational number $\alpha$. We identify $[0,1]$ with the unit circle.
	\begin{lemma}[\cite{Kh}]\label{l3} Let $\{q_n\}_{n\ge 1}$ be the sequence of the denominators of the convergents of $\alpha$ in its continued fraction expansion. For any $1\le k<q_n$,
		$$
		\|k\alpha\|_{\R/\Z}\ge \|q_{n-1}\alpha\|_{\R/\Z}>\frac{1}{2q_n}.
		$$ Thus for any $1\le k\ne k'\le q_n$, one has $$
		\|k\alpha-k'\alpha\|_{\R/\Z}>\frac{1}{2q_n}.
		$$
	\end{lemma}This separation property of irrational rotation plays an essential role in studying the exact approximation.
	We also need the uniformly distributed property of $\{k\alpha: k\ge 1\}$ for irrational $\alpha$.
	\begin{lemma}[\cite{KN}]\label{KN}
		For any irrational $\alpha$, the sequence $\{k\alpha: k\ge 1\}$ is uniformly distributed modulo 1. Equivalently, by defining the discrepancy $$
		D_n=\frac{1}{n}\sup\left\{\Big|\#\{1\le k\le n: k\alpha\in (a,b)\}-(b-a)n\Big|: 0\le a<b\le 1\right\},
		$$ one has $$
		D_n\to 0, \ {\rm{as}}\ n\to\infty.
		$$
	\end{lemma}
	
	By the definition of discrepancy, one sees that for any $0\le a <b\le 1$ and integers $m<n$ with $b-a>2D_{n-m}$,  \begin{align}\label{f0}
		\frac{(b-a)(n-m)}{2}\le\#\{m\le k< n: k\alpha\in (a,b)\}\le 2(b-a)(n-m).
	\end{align}
	Since $D_n\to 0$ as $n\to\infty$, for any given interval, the inequality (\ref{f0}) is applicable once $n-m$ is sufficiently large. It is clear that (\ref{f0}) can also be applied to the annulus. For any interval or an annulus, we use $|\cdot|$ to denote its Lebesgue measure.
	
	\subsection{Two basic propositions}\label{3-2}\
Fix $\delta_1>0$ and then $c>0$ throughtout this paper with $$
	c<\left(\frac{1}{24}\right)^2\cdot \left(1-e^{-\delta_1}\right).
	$$  For any $k\ge 1$ and $\delta>0$, define the annulus $$
	A(k)=B(k\alpha, e^{-k\delta})\setminus B(k\alpha, c e^{-k\delta}).
	$$ The annulus depends on the parameter $\delta$, and even the parameter may change at different stages later, however, we still omit this dependence and one will not be confused later.
	\begin{proposition}\label{p2}Let $\delta\ge \delta'\ge \delta_1$. Let $\ell\in \N$ be an integer with $q_l/2\le \ell<q_l$ for some $l\ge 1$ and $A(\ell)=B(\ell\alpha, e^{-\ell\delta'})\setminus B(\ell\alpha, c e^{-\ell\delta'})$ be the annulus. For all large integer $k$, there exist a collection of integers $$
		\widetilde{\D}_{k}[A(\ell)]\subset \{q_k/2\le n<q_k: n\alpha\in A(\ell)\},$$and a collection of annulus 
$$\widetilde{\A}_{k}[A(\ell)]:=\{A(n)=B(n\alpha, e^{-n\delta})\setminus B(n\alpha, c e^{-n\delta}): n\in \widetilde{\D}_{k}[A(\ell)]\}
		$$ with the following properties:
		\begin{itemize}
			\item  for the number of elements in $\widetilde{\D}_{k}[A(\ell)]$,
			$$
			\frac{1}{8}|A(\ell)|\cdot q_k\le \# \widetilde{\D}_{k}[A(\ell)]\le |A(\ell)|\cdot q_k;
			$$
			
			\item for any $n\in \widetilde{\D}_{k}[A(\ell)]$ and $x\in A(n)$, $$
			B(n\alpha, e^{-n\delta})\subset A(\ell);\quad	\text{and}\ \	\|x-m\alpha\|>c e^{-m\delta}$$
			for all $m$ with $q_l\le m<q_k$.
		\end{itemize}
	\end{proposition}
	\begin{proof}
		Applying the uniformly distributed property (see Lemma \ref{KN}) of $n\alpha$ with $q_k/2\le n<q_k$ to a smaller annulus $$
		A':=B(\ell \alpha, (1-c)e^{-\ell\delta })\setminus B(\ell \alpha, 2c e^{-\ell\delta })\subset A(\ell),
		$$ we get a collection $\D'$ of integers with the property that $$
		\D'\subset\{q_k/2\le n<q_k: n\alpha\in A'\},\ \ \#\D'\ge \frac{1-1/2}{3}|A(\ell)|\cdot q_k.
		$$
		Moreover, $B(n\alpha, e^{-n\delta})\subset A(\ell)$ for all $n\in \widetilde{\D}_{k}[A(\ell)]$ once $k$ is large, since we have already known $n\alpha\in A'$ for $n\in \D'$.
		
		To reach the last item, we need to delete those integers $n\in \D'$ such that the last item is not true.
		So we define $$
		\mathcal{F}_k=\Big\{n\in \mathcal{D}': A(n)\cap B\Big(m\alpha, c e^{-m\delta }\Big)\ne \emptyset, \ {\text{for some}}\ q_i\le m<q_k\Big\}.
		$$ The main task left is to count the cardinality of $\mathcal{F}_k$.
		
		Let $n\in \mathcal{F}_k$. By the definition of $\mathcal{F}_k$, there exists $m$ with $q_i\le m<q_k$ such that $ A(n)\cap B\big(m\al, c e^{-m\delta})\big)\ne \emptyset$. Clearly, $m\ne n$ by the definition of the annulus $A(n)$. So, together with the separation condition of $\{n\al\}$, one has\begin{align*}
			(2q_k)^{-1}<\|n\al-m\al\|_{\R/\Z}\le e^{-n\delta}+ce^{-m\delta}.
		\end{align*}
		It is clear that $e^{-n\delta}<e^{-(1/2) q_k\delta}<(4q_k)^{-1}$ once $k$ is large, so $$
		(4q_k)^{-1}\le ce^{-m\delta},\ \ {\text{and}}\
		\|n\al-m\al\|_{\R/\Z}\le (4q_k)^{-1}+ce^{-m\delta}\le 2ce^{-m\delta}.
		$$ By the triangle inequality, one can check directly that
		$$\ B(n\al, (4q_k)^{-1})\subset B(m\al, 3ce^{-m\delta}).
		$$
		
		Therefore it follows that \begin{align}\label{f1}
			\bigcup_{n\in \mathcal{F}_k}B(n\al, (4q_k)^{-1})\subset \bigcup_{q_l\le m<q_k}B(m\al, 3ce^{-m\delta}).
		\end{align}
		Note that $|A(\ell)|\geq (1-c)e^{-\ell\delta'}\ge (1-c)e^{-q_l\delta}$ and the balls in the left side of (\ref{f1}) are disjoint, so a volume argument yields that \begin{align*}
			\# \mathcal{F}_k\le 4q_k\sum_{q_l\le m<q_k}6ce^{-m\delta}\le q_k e^{-q_l\delta}\cdot \frac{24c}{1-e^{-\delta}}\le q_k|A(\ell)|\cdot \frac{1-1/2}{12},
		\end{align*} where the last inequality follows from the choice of the constant $c$.
		
		By letting $\widetilde{\D}_{k}[A(\ell)]=\D'\setminus \mathcal{F}_k$, all the required properties are satisfied.\end{proof}
	
	
	The following proposition is a simpler form of Proposition \ref{p2} which will be used only for the construction of the starting level of the Cantor set in Section \ref{3-3}.
	\begin{proposition}\label{p3} Let $\delta\ge \delta_1$. Let $I$ be an interval. For all large integer $k$, there exists a collection of integers $$
		\widetilde{\D}_{k}[I]\subset \{q_k/2\le n<q_k: n\alpha\in I\},$$and then a collection of annulus$$\widetilde{\A}_{k}[I]:=\{A(n)=B(n\alpha, e^{-n\delta})\setminus B(n\alpha, c e^{-n\delta}): n\in \widetilde{\D}[I]\}
		$$ with the following properties:
		\begin{enumerate}[label=(\roman*)]
			\item for the number of elements in $\widetilde{\D}_{k}[I]$, $$
			\frac{1}{8}|I|\cdot q_k\le \# \widetilde{\D}_{k}[I]\le q_k |I|;
			$$
			
			\item for any $n\in \widetilde{\D}_{k}[I]$, we have $
			B(n\alpha, e^{-n\delta})\subset I.
			$
		\end{enumerate}
	\end{proposition}\begin{proof}
		This follows from the uniformly distributed property of $\{n\alpha\}$.
	\end{proof}
	
	Recall a basic property on the demonimators $q_k=q_k(\alpha)$ of the convergents of $\alpha$: $$
q_{k+2}\ge q_{k+1}+q_k\ge 2q_k.
$$ Thus for any integers $n,m$ with $$
(1/2)q_i\le n<q_i, \ \ (1/2)q_j\le m<q_j, $$
if $j\ge i+2$, then $n< m$, especially $n\ne m$.
	
	\begin{proposition}\label{l3.2} Let $\delta\ge \delta'\ge \delta_1$. Let $A$ be an interval or $A=B(\ell\alpha, e^{-\ell\delta'})\setminus B(\ell\alpha, c e^{-\ell\delta'})$ an annulus with $q_l/2\le \ell <q_l$ for some $l\geq1$. Let $a$ be a small positive number with $a<\min\{2^{-10}, 2^{-10}\delta\}$ and $j\ge 1$ be an integer with \begin{align}\label{f5}
			1/2\le j\cdot 2^9a\delta^{-1}\le1.
		\end{align}
		For all large integer $k$,  there exist sub-collections $\mathcal{D}_{k+2i}[A]$ of $\widetilde{\D}_{k+2i}[A]$ for $0\le i< j$, such that \begin{enumerate}[label=(\roman*)]
			\item the collection of balls $$
			\Big\{B\Big(n\alpha, \frac{a}{\delta n}\Big): n \in\bigcup_{0\le i< j}\mathcal{D}_{k+2i}[A]\Big\}\ {\text{are disjoint}},
			$$
			\item for any $0\leq i<j$,
			$$
			\#\mathcal{D}_{k+2i}[A]\ge \frac{1}{2}\# \widetilde{\D}_{k+2i}[A]\ge \frac{1}{16}\cdot |A|\cdot q_{k+2i}.$$
		\end{enumerate}
	\end{proposition}
	\begin{proof}
Let $k$ be a large integer and $j$ be defined as above. Recall that the collections $\widetilde{\D}_{k+2i}[A]$ for $0\le i<j$ are determined by applying Proposition \ref{p2} when $A$ is an annulus or Proposition \ref{p3} when $A$ is an interval.

		Note that for each $0\le i<j$, the collection of balls $$\Big\{B(n\alpha, a /(\delta n)): n\in \widetilde{\mathcal{D}}_{k+2i}[A]\Big\}$$ are disjoint. More precisely, for any different integers $n,n'\in \widetilde{\mathcal{D}}_{k+2i}[A]$, one has $$
		q_{k+2i}/2\le n, n'<q_{k+2i}, \ {\text{and so}} \ \|n\alpha-n'\alpha\|_{\R/\Z}>(2q_{k+2i})^{-1}\geq\frac{4a}{\delta q_{k+2i}}.
		$$
		Thus we let $\D_{k}[A]=\widetilde{\D}_{k}[A]$ and call $\{n\alpha\}$ with $n\in \D_{k}[A]$ the surviving resonant points.
		
		We define $\D_{k+2i}[A]$ by induction on $i$ and the strategy is:
		to fulfill the first requirement in Proposition \ref{l3.2}, we discard those elements $m$ from $\widetilde{\D}_{k+2i}[A]$ for which the corresponding balls $B(m\al, a/(\delta m))$ will intersect the corresponding balls centered at surviving resonant points determined in the previous sub-collections $\D_{k+2t}[A]$ for $0\le t<i$.
		
		
		Assume that $\D_{k}[A], \D_{k+2}[A],\cdots, \D_{k+2(i-1)}[A]$ have been well defined.
		Let
		\begin{align*}
			\E_{k+2i}=\bigg\{m\in \widetilde{\D}_{k+2i}[A]: \ &B\Big(m\al, \frac{a}{\delta m}\Big)\cap B\Big(n\al, \frac{a}{\delta n}\Big)\ne \emptyset,\\
			& \ {\text{for some}}\ n\in\bigcup_{0\leq t<i}\D_{k+2t}[A]\, \bigg\}.
		\end{align*} Let $m\in\E_{k+2i}$. Then there exists $n\in \D_{k+2t}[A]$ for some $0\le t<i$ such that
		$$\|m\al-n\al\|_{\R/\Z}\le a \delta^{-1}m^{-1}+a\delta^{-1}n^{-1}.$$
		Since $n\ne m$, $n< m$ and $m, n<q_{k+2i}$, together with the separation condition of $\{n\al\}$, it follows that \begin{align*}
			(2q_{k+2i})^{-1}\le \|m\al-n \al\|_{\R/\Z}\le 2a\delta^{-1}n^{-1}\le 4a\delta^{-1}q_{k+2t}^{-1}.
		\end{align*} 
		Consequently, for any $y\in B(m\al, (4q_{k+2i})^{-1})$, one has \begin{align*}
			\|y-n\al\|_{\R/\Z}&\le\|y-m\al\|_{\R/\Z}+\|m\al-n \al\|_{\R/\Z}\\
			&\le (4q_{k+2i})^{-1}+4a\delta^{-1}q_{k+2t}^{-1}\le 6a\delta^{-1}q_{k+2t}^{-1}.
		\end{align*} This shows that $$B(m\al, (4q_{k+2i})^{-1})\subset B(n\al, 6a\delta^{-1}q_{k+2t}^{-1}).
		$$ In other words, \begin{align*}
			\bigcup_{m\in\E_{k+2i}}B(m, (4q_{k+2i})^{-1})\subset \bigcup_{t=0}^{i-1}\bigcup_{n\in \D_{k+2t}[A]}B(n \al, 6a\delta^{-1}q_{k+2t}^{-1}).
		\end{align*} Since the left union is disjoint, a volume argument shows that \begin{align*}
			(2q_{k+2i})^{-1} \cdot\# \E_{k+2i} \le \sum_{t=0}^{i-1}q_{k+2t}\cdot |A|\cdot 12a \delta^{-1}q_{k+2t}^{-1}=12ia\delta^{-1}|A|.
		\end{align*} Then one has, \begin{align*}
			\# \E_{k+2i} &\le 24ia\delta^{-1} |A|\cdot q_{k+2i}\le \frac{1}{16}q_{k+2i}|A|\le \frac{1}{2}\#\widetilde{\D}_{k+2i}[A].
		\end{align*} where the second inequality follows from (\ref{f5}) on the choice of $j$. Now by letting $\D_{k+2i}[A]=\widetilde{\D}_{k+2i}[A]\setminus \E_{k+2i}$, we get the desired collection of integers.\end{proof}
	
	%
%

\subsection{The Cantor subset $C(\delta)$}\label{3-3} Let $\Delta=\{\delta_k:k\geq1\}$ be a non-decreasing sequence of positive
numbers starting from $\delta_1$ which has already been fixed before Proposition \ref{p2}.
Equipped with Proposition \ref{p2} and Proposition \ref{l3.2},
we begin the Cantor subset construction.


{\bf The first level of the Cantor subset}\

Fix a small positive number $a_0<\min\{2^{-10}, 2^{-10}\delta_1\}$ arbitrarily. Let $h_1$ be an integer such that for all $n\ge (1/2)q_{h_1}$,
\begin{align*}
\frac{2^{-14}\cdot e^{-n\delta_{1}}}{(n\delta_{{1}})^{-1}}<\min\{2^{-16}a_0,\  2^{-10}\delta_2\}.\end{align*}
Applying Proposition \ref{p3} and then Proposition \ref{l3.2} to the interval $I=[0,1]$ with $a=a_0$, $\delta=\delta_1$, then there exist large integers $k_1\ge h_1,\ j_1\in\N$ such that \begin{align}\label{ff4}
1/2\leq j_1\cdot 2^9a_0\delta_{1}^{-1}\leq 1,
\end{align}
and collections of resonant points $\D_{k_1+2i}[I]$ for $0\le i< j_1$ satisfying the requirements in Proposition \ref{l3.2}.

At this stage, for each $n\in \D_{k_1+2i}[I]$, the annulus is defined as $$A(n)=B(n\al, e^{-n\delta_{1}})\setminus B(n\al, ce^{-n\delta_{1}}).$$
The first level is defined as $$
F_1=\bigcup_{n\in\F_1}A(n),\ \ {\text{where}}\ \F_1=\bigcup_{i=0}^{j_1-1}\D_{k_1+2i}[I].
$$
It is clear that $F_1$ consists of a collection of disjoint annulus by Proposition \ref{l3.2}, since $$
e^{-n\delta_{1}}\le \frac{a}{\delta_{1} n}, \quad{\text{and then}}\  A(n)\subset B(n\al, \frac{a}{\delta_{1} n}), \ {\text{for all}}\ n\in \F_1.
$$

%
%

	
	
	
	\medskip

		{\bf The general level of the Cantor subset}\
		
		We adopt the convention $F_0=I=[0,1]$. Let $t\geq2$. Assume that the Cantor subset has been defined up to level $t-1$ where each level consists of a collection of disjoint annulus.
		Now we define the $t$-th level which contains a collection of sub-levels.
		
		Fix an element $n_{t-1}\in \F_{t-1}$ and then the corresponding annulus $$A(n_{t-1})=B(n_{t-1}\al, e^{-n_{t-1}\delta_{t-1}})\setminus B(n_{t-1}\al, ce^{-n_{t-1}\delta_{t-1}}).$$ By the inductive process, we know $n_{t-1}\in \mathcal{D}_{{k_{t-1}+2i_{t-1}}}[A(n_{t-2})]$ for some $A(n_{t-2})\in \F_{t-2}$ and integers $k_{t-1}$, $i_{t-1}$ and $j_{t-1}$ with $0\le i_{t-1}< j_{t-1}$ and
		\begin{align}\label{ff5}a_{t-1}:=\frac{2^{-14}\cdot e^{-n_{t-1}\delta_{{t-1}}}}{(n_{t-1}\delta_{{t-1}})^{-1}}<\min\{2^{-10}, 2^{-10}\delta_t\}.
		\end{align}
Let $h_t$ be an integer such that for all $n\ge (1/2)q_{h_t}$,
\begin{align}\label{fff1}
\frac{2^{-14}\cdot e^{-n\delta_{t}}}{(n\delta_{{t}})^{-1}}<\min\{2^{-16}a_{t-1},\  2^{-10}\delta_{t+1}\}.\end{align}
Then choose a large integer $k_t\ge h_t$ and let  $j_t\in\N$ be an integer with
		\begin{align}\label{ff41}
			1/2\le j_t\cdot2^{9} a_{t-1}\delta_{t}^{-1}\le 1.
		\end{align}
		
		Applying Proposition \ref{l3.2} to the annulus $A(n_{t-1})$ with $a=a_{t-1}$ and $\delta=\delta_t$, there exist the collections  $\D_{k_t+2i}[A(n_{t-1})]$ for
$0\le i< j_t$ satisfying the requirements in Proposition \ref{l3.2}. Note that $j_t$ depends on $n_{t-1}$, 
however we donot emphasis this dependence so omit it in notation.
		
At this stage, for each $n\in \D_{k_t+2i}[I]$, the corresponding annulus is defined as $$A(n)=B(n\al, e^{-n\delta_{t}})\setminus B(n\al, ce^{-n\delta_{t}}).$$ Then a local $t$-th level is defined as $$
		F_t(n_{t-1})=\bigcup_{n\in\F_t(n_{t-1})}A(n),\ \ {\text{by letting}}\ \F_t(n_{t-1})=\bigcup_{i=0}^{j_t-1}\D_{k_t+2i}[A(n_{t-1})],
		$$ and the $t$-th level is defined as $$
		F_t=\bigcup_{n_{t-1}\in \F_{t-1}}\bigcup_{n\in \F_t(n_{t-1})}A(n),\ {\text{and letting}}\ \F_t=\bigcup_{n_{t-1}\in \F_{t-1}}\F_t(n_{t-1}).
		$$
		
		For convenience, for all $t\in\N$, we list the properties shared by the quantities appearing in the $t$-th local level: let $n_{t-1}\in \F_{t-1}$.\begin{itemize}
			\item For each $0\le i< j_t$, by Proposition \ref{l3.2} (ii), \begin{align}\label{f7}
				\frac{1}{16}q_{k_t+2i}\cdot |A(n_{t-1})|\le \# \D_{k_t+2i}[A(n_{t-1})]\le   q_{k_t+2i}\cdot |A(n_{t-1})|,
			\end{align}
			
			\item For any $n \ne n '\in \F_t(n_{t-1})$, by Proposition \ref{l3.2} (i), \begin{align}\label{f8}
				B\left(n\al, \frac{a_{t-1}}{n\delta_{t}}\right)\cap B\left(n'\al, \frac{a_{t-1}}{n'\delta_{t}}\right)=\emptyset.
			\end{align}
			
			%
		\end{itemize}
		
		Similarly, all the annulus $A(n)$ in $\F_t(n_{t-1})$ are disjoint and contained in $A(n_{t-1})$. In view of the disjointness of the annulus
		in $F_{t-1}$, the $t$-th level $F_t$ consists of a collection of disjoint annulus and there is a nested structure between $F_{t-1}$ and $F_t$.
		

The desired Cantor set is defined as $$
C:=\bigcap_{t=1}^{\infty}F_t=\bigcap_{t=1}^{\infty}\bigcup_{n\in \F_t}A(n).
$$ By specifying the non-decreasing sequence $\Delta=\{\delta_k:k\geq1\}$ as $$
\delta_k=\min\{\delta, \log\log k\} \ {\text{when $k\ge e^3$,\quad and }}  \delta_k=\min\{\delta, 1\} \ {\text{for other $k$}},
$$ we have the following lemma.
		\begin{lemma}\label{sub-Cantor}
			For any $0<\delta\leq\infty$, we have $$C\subset D(\delta).$$
		\end{lemma}
		\begin{proof}
			Let $x\in C$. Bear in mind the nested structure of $\{F_t: t\ge 1\}$ and the disjointness of the annulus in $F_t$ for each $t\ge 1$. There is a sequence of integers $\{n_t\}$ with $n_t\in \F_t$ for all $t\ge 1$ and correspondingly two sequences of integers $\{k_t\}, \{i_t\}$ such that for each $t\ge 1$, $$
			x\in A(n_t),\quad A(n_t)\subset A(n_{t-1}), \ q_{k_t+2i_t}/2\le n_t\le q_{k_t+2i_t}.
			$$ We check that $x\in D(\delta)$. Observe that we have the following facts:
			\begin{itemize}
				\item For each $t\ge 1$, $$
				x\in A(n_t)\subset B(n_t\al ,e^{-\delta_t n_t}).
				$$
				
				\item For each $n\ge q_{k_1+2i_1}$, there exists $t\ge 2$ such that $$
				q_{k_{t-1}+2i_{t-1}}\le n<q_{k_t+2i_t}.
				$$
			\end{itemize}
			Remind that $x\in A(n_{t})$ and $n_t\in \D_{k_t+2i_t}[A(n_{t-1})]$. Since $\mathcal{D}_{k_t+2i_t}[A(n_{t-1})]$ is a sub-collection of $\widetilde{\D}_{k_t+2i_t}[A(n_{t-1})]$, and by Proposition \ref{p2} about the property of $\widetilde{\D}_{k_t+2i_t}[A(n_{t-1})]$, one has that $$
			\|x-m\alpha\|>c e^{-m\delta_t}\quad\text{for all $m$ with $q_{k_{t-1}+2i_{t-1}}\le m<q_{k_t+2i_t}$.}
			$$
			This shows that $x\in D(\delta)$.
			%
		\end{proof}
		\subsection{Mass distribution on $C$}\label{3-4}
		
		For any $0<\delta\leq\infty$, we define a probability measure supported on $C$. So, let $\mu([0,1])=1$. Remind that the annulus in the $t$-th level of the Cantor set is defined as $$
A(n)=B(n\al, e^{-n\delta_{t}})\setminus B(n\al, ce^{-n\delta_{t}}),\ \ {\text{for all}}\ n\in \F_t.
$$ 

The measure of $\mu$ on the annulus of the first level is defined as follows: for $n_1\in \F_1$, we define $$
		\mu(A(n_1))=\frac{(\delta_1n_1)^{-1}}{\sum_{0\le i< j_1}\sum_{n\in \D_{k_1+2i}}(\delta_1n)^{-1}}.
		$$ By letting $t=1$ in (\ref{f7}), it follows that \begin{align*}
			\sum_{0\le i< j_1}\sum_{n\in \D_{k_1+2i}}(\delta_1n)^{-1}&\ge \sum_{0\le i< j_1}\frac{1}{16}q_{k_1+2i}\cdot (\delta_1q_{k_1+2i})^{-1}=\frac{j_1}{2^4\delta_1}.
		\end{align*} Thus by \eqref{ff4}, we have
		$$
		\mu(A(n_1))\le (\delta_1n)^{-1}\cdot \frac{2^4\delta_1}{j_1}\le 2^{14}a_0 (n_1\delta_{1})^{-1}.
		$$
		
		Now we define $\mu$ on the annulus in every local $t$-th level inductively.
		Then for each $n_t\in \F_t(n_{t-1})$, define\begin{align*}
			\mu(A(n_t))&=\frac{(\delta_tn_t)^{-1}}{\sum_{n\in \F_t(n_{t-1})}(\delta_tn)^{-1}}\cdot\mu(A(n_{t-1}))\\
 &=\frac{(\delta_tn_t)^{-1}}{\sum_{0\le i< j_t}\sum_{n\in \D_{k_t+2i}[A(n_{t-1})]}(\delta_tn)^{-1}}\cdot \mu(A(n_{t-1})).
		\end{align*}
		
		It is clear that $\mu$ satisfies the consistency property, so by Kolmogorov's extension theorem, it can be extended uniquely into a probability measure supported on $C$.
		
		Let $t\geq2$. Assume that for any $n_{t-1}\in \F_{t-1}$, we have proved\begin{align}\label{ff7}
			\mu(A(n_{t-1}))\le 2^{14}a_0 (n_{t-1}\delta_{{t-1}})^{-1}.
		\end{align}
		We will show this is also true for $n_t\in \F_t$.
		
		It follows from \eqref{f7} that \begin{align*}
\sum_{0\le i<j_t}\sum_{n\in \D_{k_t+2i}[A(n_{t-1})]}(\delta_t n)^{-1}
			&\ge \sum_{0\le i< j_t}\frac{1}{16}\cdot q_{k_t+2i}|A(n_{t-1})|\cdot (\delta_t q_{k_t+2i})^{-1}\\
			&\ge \frac{j_t}{2^4\delta_t}\cdot e^{-n_{t-1}\delta_{{t-1}}}\ge (\delta_{{t-1}}n_{t-1})^{-1},
		\end{align*} where for the last inequality, we use the choice of $a_{t-1}$ in (\ref{ff5}) and $j_t$ in (\ref{ff41}). Thus by (\ref{ff7}), one has$$
		\mu(A(n_{t}))\le \frac{(\delta_{t}n_t)^{-1}}{(\delta_{{t-1}}n_{t-1})^{-1}}\cdot 2^{14}a_0 (n_{t-1}\delta_{t-1})^{-1}=2^{14}a_0 (\delta_{{t}}n_t)^{-1}.
		$$

		\subsection{Hausdorff measure of $C$}\label{3-5}
		
		Recall $\omega_1(r)=(-\log r)^{-1}$. In this subsection, we will show that
		\begin{proposition}\label{Cantor-delta}
			For any $0<\delta\leq\infty$, we have $
			\mathcal{H}^{\omega_1}(C)=\infty.$
		\end{proposition}
		\begin{proof}
			We use the mass distribution principle (Lemma \ref{Fal}) to conclude the Hausdorff measure of $C$ for $0<\delta\leq\infty$ by showing that for all $x\in C$ and $r$ small, \begin{align}\label{ff9}
				\mu(B(x,r))\le 2^{32}a_0\cdot \omega_1(r).
			\end{align}
			
			Fix a ball $B={B}(x,r)$ with $x\in C$ and $r$ small enough such that $B(x,r)$ can intersect only one annulus in the first level of $C$. If the ball $B(x,r)$ can intersect only one annulus in $F_t$ for all $t\ge 1$, it follows that $$
			\mu(B(x,r))\le \mu(A(n_t))\le 2^{14}a_0(n_t\delta_t)^{-1}\to 0
			$$ as  $t\to \infty.$  So (\ref{ff9}) is true trivially. Thus in the following we assume that the ball $B(x,r)$ can intersect at least two annulus in $F_t$ for some $t\ge 1$.
			
			Let $t$ be the smallest integer such that the ball $B(x,r)$ can intersect at least two annulus in the $t$-th level of $C$. Let $n_{t-1}\in \F_{t-1}$ be the unique integer such that $B(x,r)\cap A(n_{t-1})\ne \emptyset$. We can further assume that $r<e^{-n_{t-1}\delta_{{t-1}}}$, since, otherwise, one has $$
			\mu(B(x,r))\le \mu(A(n_{t-1}))\le 2^{14}a_0(n_{t-1}\delta_{{t-1}})^{-1}\le 2^{14}a_0(-\log r)^{-1}.
			$$
			
			By the uniqueness of $n_{t-1}$, all the annulus in the $t$-th level of $C$ for which the ball $B(x,r)$ can intersect are contained in$$
			\left\{A(n): n\in \F_{t}(n_{t-1})\right\}=\Big\{A(n): n\in \bigcup_{0\le i< j_t}\mathcal{D}_{k_t+2i}[A(n_{t-1})]\Big\}.
			$$
			By \eqref{f8}, the balls $\{B(n\alpha, a_{t-1}/(n\delta_{t}): n\in \F_t(n_{t-1})\}$ are disjoint and by  (\ref{fff1}), $A(n) \subset B(n\alpha, a_{t-1} /(n\delta_{t}))$. Let $n\in \F_t(n_{t-1})$ be such that $B(x,r)\cap A(n)\ne \emptyset$. Since $B(x,r)$ can intersect at least two annulus in the $t$-th level $F_t$, it follows that $$
			B(x,r)\cap B(n\al, e^{-n\delta_t})\ne \emptyset, \ B(x,r)\setminus B(n\al, a_{t-1}(n\delta_{t})^{-1})\ne \emptyset.
			$$
			Thus $$
			2r\ge a_{t-1}(n\delta_{t})^{-1}-e^{-n\delta_t}.
			$$ By the inequality in (\ref{fff1}) again, one has
			$$
			2r\ge \frac{1}{2}\cdot a_{t-1}\cdot (n\delta_{t})^{-1},
			$$
			and  thus $$
			\bigcup_{n\in \F_t(n_{t-1}): A(n)\cap B(x,r)\ne \emptyset}B\left(n\al, a_{t-1}(n\delta_{t})^{-1}\right)\subset B(x, 9r).
			$$
			
			By the definition of  the measure $\mu$, it follows that
			\begin{align*}
				\mu(B(x,r))&\le \sum_{n_t\in \F_t(n_{t-1}): A(n_t)\cap B(x,r)\ne \emptyset}\mu(A(n_t))\\
				&= \sum_{\substack{n_t\in \F_t(n_{t-1}),\\
						A(n_t)\cap B(x,r)\ne \emptyset}}\frac{(\delta_tn_t)^{-1}}{\sum\limits_{0\le i<j_t}\sum\limits_{n\in \D_{k_t+2i}[A(n_{t-1})]}(\delta_tn)^{-1}}\cdot \mu(A(n_{t-1}))\\
				&=\sum_{\substack{n_t\in \F_t(n_{t-1}),\\
						 A(n_t)\cap B(x,r)\ne \emptyset}}\frac{ 2a_{t-1} (n_t\delta_{t})^{-1}}{2a_{t-1}\delta_{t}^{-1}\sum\limits_{0\le i< j_t}\sum\limits_{n\in \D_{k_t+2i}[A(n_{t-1})]}n^{-1}}\cdot \mu(A(n_{t-1})).
			\end{align*}
			For the denominator, by \eqref{f7} and (\ref{ff41}), one has\begin{align*}
				2a_{t-1}\delta_{t}^{-1}\sum_{0\le i< j_t}\sum_{n\in \D_{k_t+2i}[A(n_{t-1})]}n^{-1}&\ge \frac{a_{t-1}j_t\delta_{t}^{-1}}{2^3}\cdot |A(n_{t-1})|\ge 2^{-13}e^{-n_{t-1}\delta_{{t-1}}}.
			\end{align*}
			For the numerator, \begin{align*}
\sum_{\substack{n_t\in \F_t(n_{t-1}),\\
		A(n_t)\cap B(x,r)\ne \emptyset}}2a_{t-1}(n_t\delta_{t})^{-1}&\le \sum_{\substack{n_t\in \F_t(n_{t-1}),\\
		A(n_t)\cap B(x,r)\ne \emptyset}} \left|B\Big(n_t\al, a_{t-1}(n_t\delta_{t})^{-1}\Big)\right|\\
				&\le |B(x,9r)|\le 18 r\le2^{5}r.
			\end{align*}
			Therefore, together with (\ref{ff7}), one has
			\begin{align*}
				\mu(B(x,r))
				&\le \frac{2^{18}r}{e^{-n_{t-1}\delta_{{t-1}}}}\cdot 2^{14}a_0(n_{t-1}\delta_{{t-1}})^{-1}\\
				&=2^{32}a_0\cdot \frac{r}{(-\log r)^{-1}}\cdot \frac{(-\log e^{-n_{t-1}\delta_{{t-1}}})^{-1}}{e^{-n_{t-1}\delta_{{t-1}}}}\cdot (-\log r)^{-1}\le \frac{2^{32}a_0}{-\log r},
			\end{align*} where for the last inequality we use the fact $$
			\frac{(-\log r)^{-1}}{r}\ {\text{is increasing as $r\to 0$ and}}\ r<e^{-n_{t-1}\delta_{{t-1}}}. $$

			In a summary, we have shown that for all $x\in C$ and $r$ small, $$
			\mu(B(x,r))\le 2^{32}a_0\cdot (-\log r)^{-1}.
			$$
			Then an application of Lemma \ref{Fal} yields that $$
			\mathcal{H}^{\omega_1}(C)\ge 2^{-32}a_0^{-1}.
			$$ The desired result follows by the arbitrariness of $a_0$.
		\end{proof}

		\subsection{Proof of Theorem \ref{F}}\label{3-6}
	
%
			For any $0<\eta<\delta$, define $$
			B(\eta)=\left\{x\in [0,1]: \|x-k\alpha\|< e^{-|k|\eta}\ {\text{for infinitely many}}\ |k|\in \mathbb{N}\right\}.
			$$
			Then it is clear that $D(\delta)\subset B(\eta)$.
A simple Borel-Cantelli argument yields the result that for any $s>1$ and $0<\eta<\infty$, $\mathcal{H}^{\omega_s}(B(\eta))=0,$ and so $\mathcal{H}^{\omega_s}(D(\delta))=0.$
			
			
	For $s\le 1$, since	$C\subset D(\delta)$ by Lemma \ref{sub-Cantor}	for any $0<\delta\leq\infty$, then  by Proposition \ref{Cantor-delta} one has
			\begin{equation*}
				\mathcal{H}^{\omega_1}(D(\delta))\geq\mathcal{H}^{\omega_1}(C)=\infty.
			\end{equation*}
			Since $\omega_s\geq\omega_1$ for any $s\leq1$, then $\mathcal{H}^{\omega_s}(D(\delta))\geq\mathcal{H}^{\omega_1}(D(\delta))=\infty$. \qed

		\subsection{Proof of Theorem \ref{DtoF}}

		By \eqref{Xmu-beta}, Theorem \ref{DtoF} just follows from Proposition \ref{DtoF'} and Theorem \ref{F}.\qed

        \subsection{Proof of Theorem \ref{thm:stdhausdorff}}If  $0 < \lambda < 1$ and $\alpha \in DC$, the spectral measure is purely absolutely continuous for all $\theta$ \cite{avila2008absolutely}. Since the spectral measure is absolutely continuous with respect to the Lebesgue measure, the set where $\beta = 1$ must have full Hausdorff dimension. Consequently, Theorem \ref{thm:stdhausdorff} follows directly from Theorem \ref{DtoF}. \qed
	
\section*{Acknowledgements} 
 This work was partially supported by National Key R\&D Program of China (2020YFA071300, 2024YFA1013700). J. Cao was supported by Nankai Zhide Foundation. X. Li was supported by NSFC grant (123B2005) and an AMS-Simons Travel Grant. B. Wang was support by NSFC grant (12331005). Q. Zhou was supported by NSFC grant (12531006) and  Nankai Zhide Foundation.

 \section*{Data Availability Statement}
Data sharing is not applicable to this article as no datasets were generated or analyzed.

\section*{Conflict of Interest}
The authors declare that there are no conflicts of interest regarding the publication of this work.

		\bibliographystyle{siam}
		\bibliography{Appendix-1.bib}
	\end{document}